\documentclass[algorithms,conceptpaper,accept,moreauthors,pdftex,12pt,a4paper]{mdpi} 
\lastpage{13}
\doinum{10.3390/------}
\pubvolume{xx}
\pubyear{2015}
\history{Received: xx / Accepted: xx / Published: xx}
\pdfoutput=1

\usepackage{setspace}
\usepackage{graphicx}
\usepackage{subfigure}
\usepackage{url}
\usepackage{verbatim}
\usepackage{algorithm}
\usepackage{algorithmic}
\usepackage{array}
\usepackage{amsmath}
\usepackage{amssymb}
\usepackage{latexsym}
\usepackage{algorithm,algorithmic}
\usepackage{lmodern}
\theoremstyle{mdpi}
\newcounter{thm}
\newcounter{ex}
\newcounter{re}
\newtheorem{Theorem}[thm]{Theorem}

\newtheorem{Definition}[thm]{Definition}

\Title{Nearest Neighbor search in Complex Network for Community Detection}

\Author{Suman Saha $^{1,}$* and S.P. Ghrera $^{1}$}

\address{%
$^{1}$ Institute 1, Jaypee University of Information Technology 1, Waknaghat, Solan, Himachal, India}

\corres{suman.saha@juit.ac.in}

\abstract{Nearest neighbor search is a basic computational tool used extensively in almost research domains of computer science specially when dealing with large amount of data. However, the use of nearest neighbor search is restricted for the purpose of algorithmic development by the existence of the notion of nearness among the data points. The recent trend of research is on large, complex networks and their structural analysis, 
where nodes represent entities and edges represent any kind of relation between entities. Community detection in complex network is an important problem of much interest.  In general, a community detection algorithm represents an objective function and captures the communities by optimizing it to extract the interesting communities for the user.
In this article, we have studied the nearest neighbor search problem in complex network via the development of a suitable notion of nearness. Initially, we have studied and analyzed the exact nearest neighbor search using metric tree on proposed metric space constructed from complex network. After, the approximate nearest neighbor search problem is studied using locality sensitive hashing. 
For evaluation of the proposed nearest neighbor search on complex network we applied it in community detection problem. The results obtained using our methods are very competitive with most of the well known algorithms exists in the literature and this is verified on collection of real networks. On the other-hand, it can be observed that time taken by our algorithm is quite less compared to popular methods.
}


\begin{document}


\section{Introduction}
\label{nn-int}
Nearest neighbor (NN) search is an important computational primitive for structural
analysis of data and other query retrieval purposes. NN search is very useful for dealing
with massive datasets, but it suffer with "curse of dimensionality"\cite{UHLMANN1991, VIDALRUIZ1986}. 
However, some 
recent surge of results show that it is also very efficeint for high dimensional data
provided a suitable space partitioning data structure is used, like, kd-tree, quad-tree, R-tree, metric-tree
and locality sensitive hashing\cite{Panigrahy:2006, Indyk:1998, Gionis:1999, Dasgupta:2008}. 
Some of these data structures also support approximate
nearest neighbor search which hardly made any degradation of results whereas saves lot of
computational times. In the NN-search problem, the goal is to pre-process a set of data points, so
that later, given a query point, one can find efficiently the data point nearest to the
query point on some metric space of consideration. NN search has many applications 
in data processing and analysis. For instance, information retrieval, searching image databases, finding duplicate pages, compression, and many others.
To represent the objects and the similarity measures, one often uses geometric notions of nearness
\cite{Akoglu:2014, Liu04}.

One important research direction of recent interest is to extract  network communities in large real graphs such as social networks, web, collaboration networks and bio-networks \cite{Freeman78centralityin, CarringtonScott:2005, newman2003structure, radicchi2004}. 
The availability of large, detailed datasets
representing such networks has stimulated extensive study of their basic properties, and the 
identification of hirarchical structural features\cite{newman2003structure,Fortunato10}. 
Other than graphs, the complex networks are characterized by small average path length and high clustering coefficient.
A network community (also known as a module or cluster) is typically a group of nodes with more interconnections among its members than the remaining part of the network \cite{weiss, Schaeffer200727, Fortunato10}.
To extract such group of nodes from a network one generally selects an objective function that captures the possible communities as a set of nodes with better internal connectivity than external
\cite{NewGir04,Luxburg07}.
However, very less research is done for network community detection which tries to develop nearness among the 
nodes of a complex network and use nearest neighbor search for partitioning the 
network\cite{Pons04, duch-2005, Chakrabarti04a, MacropolS10, LevoratoP11, BrandesGW03,bullmore2009complex}.
Complex networks are characterized by small average path length and high clustering coefficient the way the metric is defined should be able to capture the crucial properties of complex networks. Therefore, we need to create the metric very carefully so that it can explore the underlying community structure of the real life networks\cite{SahaG15}.

In this work, we have developed the notion of nearness among the nodes of the network using some new matrices derived from modified adjacency matrix of the graph which is flexible over the networks and can be tuned to enhance the structural properties of the network required for community detection. 
The main contributions of this work include:

\begin{itemize}
\item Define a metric on complex network suitable for community detection

\item Choice of data structure  and approximations for NN computation 
(M-tree, LSH) 

\item Community detection algorithm developed using NN search

\item Experiments on real networks
\end{itemize}

The rest of this paper is organized as follows:In Section \ref{nn-cn} several definitions and challanges relevent to nearest neighbor search in complex network are discussed. Section \ref{nn-near} describes the notion of nearness in complex network and developed a method to represent a complex network as points of a metric space. Section \ref{nn-mtree} describes the problem of nearest neighbor search over complex network and the use of metric tree data structure in this regard.
In Section \ref{nn-approx}, the problem approximate nearest neighbor search on complex network is discussed with a newly developed locality sensitive hashing method. Network community detection using exaxt and approximate nearest neighbor search is formulated and several possible solutions are presented in Section \ref{nn-comm}, also, the initialization procedures, termination criteria, convergence are discussed in detail.  
The results of comparison between community detection algorithms are illustrated in Section \ref{nn-res}. The computational aspects of the proposed framework are also discussed in this section. 

\section{NNSCN: Nearest neighbor search on complex network}
\label{nn-cn}

The nearest-neighbor searching problem is to find the nearest points in a $D$ dimensional dataset $X \subset R^D$ containing $n$ points to a query point $q \in R^D$ , usually in a metric space. It
has applications in a wide range of real-world settings, in particular pattern recognition,
machine learning and database querying to name a few. Several effective methods exist for this
problem when the dimension $D$ is small, such as Voronoi diagrams,
however, kd-trees and metric trees are common when the dimension is high. 
Many methods with different approach are developed for searching data
and finding the nearest point. Searching the nearest neighbor in different studies 
are presented by different names such as post office problem, proximity search, closest point
search, best match file searching problem, index for similarity
search, vector quantization encoder, the light-bulb problem
and etc.. The solutions for the Nearest Neighbor Search
(NNS) problem usually have two parts: nearness determination in the data and
and algorithmic developments. In most the NNS
algorithms, the main framework is based on four fundamental
algorithmic ideas: Branch-and-bound, Walks, Mapping-based
techniques and Epsilon nets. There are thousands of possible
framework variations and any practical application can lead to
its unique problem formalization such as pattern recognition,
searching in multimedia data, data compression,
computational statistics, information retrieval, databases and
data mining, machine learning, algorithmic theory,
computational geometry, recommendation systems and etc.

\subsection{NN problem definition on complex network}
A NNS problem defined in a metric space is defined below.
\begin{Definition}[Metric space] Given a set $S$ of points and $d$ as
a function to compute the distance between two points. Pair
$(S, d)$ distinguished metric space if $d$ satisfies
 reflexivity, non-negativity, symmetry and triangle
inequality.
\end{Definition}
Non-metric space data are indexed by special data structures
in non-metric spaces and then searching is done on these
indexes. A few efficient methods exist for searching in non-
metric space that in most of them, non-metric space is
converted to metric space. 
The focus of this paper is on the problems defined on a metric space. In
a more detailed classification, NNS problems can be defined
in Euclidean space as follow:

\begin{Definition}[Nearest neighbor search] Given a set $S$ of points in a $D$ dimensional 
space, construct a data structure
which given any query point finds the point in $S$ with
the smallest distance with $q$.
\end{Definition}

This definition for a small dataset with low dimension has sub
linear (or even logarithmic) query time, but for massive
dataset with high dimension is exponential. Fortunately,
some little approximation can decrease the exponential complexity into
polynomial time.

Approximate NNS is defined as:

\begin{Definition}[Approximate nearest neighbor] Given a set $S$
of Points in a $D$-dimensional space, construct a
data structure which given any query point, reports
any point within distance at most $c$ times the nearest distance from $q$.
\end{Definition}
The first requirement in order to search in a metric space is
the existence of a formula to calculate the distance
between each pair of objects in $S$. Different metric distance
functions can be defined depending on the search space of consideration. 
A NN query on a complex network G, consists of a
source node $s$ and a metric function $d(x, y)$.  
This computations depends on the dimension of the instance and face "curse of dimensionality" problem.
The computation can be reduced drastically if instead of computing the exact nearest neighbor we compute the approximate nearest neighbor.

\section{Notion of nearness in complex network }
\label{nn-near}
The notion of nearness among the nodes of a graph are used in several purposes in the history fo literature of graph theory. Most of the time the shortest path and edge connectivity are popular choice to describe nearness of nodes. However, that edges do not give the true measure of network connectivity (proof by kleinbarg). 
The notion of network connectivity some times generalized to be the number of paths, of any length, that
exist between two nodes. This measure, called influence
by sociologists, because it measures the ability of one node to affect another, gives a better measure of connectivity between nodes of real life graphs / complex networks. Beside discovering natural groups within a network, the influence metric
can also help identify the weak ties who bridge different communities. 
Research in this direction gained special attention in the domain of complex network analysis, some of them along with the one proposed in this article are discussed in the following subsections.
\subsection{Nearness literature}
\label{near-lit}

Methods based on node neighborhoods. For a node $x$, let $N(x)$ denote the set of neighbors
of $x$ in a graph $G(V, E)$ . A number of approaches are based on the idea that two nodes $x$ and 
$y$ are more likely to be affected by one another if their sets of neighbors 
$N(x)$ and $N(y)$ have large overlap. 

Common neighbors: The most direct implementation of this idea for nearness computation is to define
$d(x, y) := |N(x) \cap N(y)|$, the number of neighbors that $x$ and $y$ have in common. 

Jaccard coefficient: The Jaccard coefficient, a commonly used similarity
metric,  measures the probability that both $x$ and $y$ have a feature $f$ ,
for a randomly selected feature $f$ that either $x$ or $y$ has. If we take features here to be neighbors
in $G(V, E)$ , this leads to the measure $d(x, y) := |N(x) \cap N(y)|/|N(x) \cup N(y)|$. 

Preferential attachment:The probability that a new edge involves node $x$ is proportional
to $|N(x)|$, the current number of neighbors of $x$. The probability of co-authorship of $x$ and $y$ is
correlated with the product of the number of collaborators of $x$ and $y$. This corresponds to the
measure $d(x, y) := |N(x)| \times |N(y)|$.

Katz measure: This measure directly sums over the collection of paths, exponentially
damped by length to count short paths more heavily. This leads to the measure
$d(x, y) := \beta \times |paths(x,y) | $
where $paths(x,y)$ is the set of all length paths from $x$ to $y$. ($\beta$ determines the path size, since paths of length three or more contribute very little to the summation.)

Hitting time and PageRank: A random walk on $G$ starts at a node $x$, and
iteratively moves to a neighbor of $x$ chosen uniformly at random. The hitting time $Hx,y$ from $x$ to $y$
is the expected number of steps required for a random walk starting at $x$ to reach $y$. Since the hitting
time is not in general symmetric, it is also natural to consider the commute time $C(x,y) := H(x,y) + H(y,x)$.
Both of these measures serve as natural proximity measures, and hence (negated) can be used as
$d(x, y)$. Random resets form the basis of the PageRank measure for Web pages, and we can adapt it
for link prediction as follows: Define $d(x, y)$ to be the
stationary probability of $y$ in a random walk that returns to $x$ with probability $\alpha$ each step, moving
to a random neighbor with probability $1 - \alpha$.

Most of the methods are developed for different types of problems like information retrieval, ranking, prediction e.t.c. and developed for general graphs. In this article we studied a measure specially designed for complex network and discussed in the next subsection 
\subsection{Proposed metric on complex network}
\label{nn-metric}
In this section we have demonstrated the procedure to transform a graph into points of a metric space and developed the methods of community detection with the help of metric defined for pair of points. We have also studied and analyzed the community structure of the network therein.  

The nodes of the graph do not lie on a metric space. The standard Euclidean distance and spherical distance define over the adjacency or Laplacian matrices above failed to capture similarity information among the nodes of a complex network. On the other-hand, the algorithms developed based on shortest path or Jaccard similarity are computationally inefficient and have less success in terms of standard evaluation criteria(like, conductance and modularity). 

In this work, we have tried to develop the notion of similarity among the nodes using some new matrices derived from adjacency matrix and degree matrix of the graph.
Let $A$ be the adjacency matrix and $D$ the degree matrix of the graph $G= (V, E)$.  The Laplacian $L = D -A$. We have defined two diagonal matrix of same size $D (\lambda)$ and $D (\lambda_{x})$ where $\lambda$ is a parameter determine from the given graph and can be optimized from the optimization criteria of the problem under consideration.
In $D(\lambda)$ a fixed optimally determine value is used in the diagonal entries of the matrix $D$ and in $D(\lambda_{x})$ a variable value also optimally determine  is used in the diagonal entries of the matrix $D$.
The similarities are defined on matrices $L_1$ and $L_2$, where $L_1 = D(\lambda) +A$ and $L_2 = D(\lambda_{x}) +A$ respectively as spherical similarity among the rows and determine by applying a concave function $\phi$ over the standard notions of similarities like, Pearson coefficient($\sigma_{PC}$), Spacerman coefficient($\sigma_{SC}$) or Cosine similarity($\sigma_{CS}$). $\phi(\sigma)()$ must be chosen using the chord condition to obtain a metric.

In this subsection we have demonstrated the algorithm to convert the nodes of the graph to the points of a metric space preserving the community structure of the graph.
The algorithm depends on the sub modules 1) construction of $L_x$ ($L_1$ or $L_2$) and 2) obtaining a structure preserving distance function. The algorithm works as picking pair of nodes from $L_x$ and and computing distance defined in the second module.

\subsubsection{$L_x$ construction} The $L_1$ is defined as $L_1 = D(\lambda) +A$, where $A$ is the adjacency matrix of the given network and $D(\lambda)$ is a diagonal matrix of same size with diagonal values equal to a non negative constant $\lambda$.

The $L_2$ is defined as $L_2 = D(\lambda_{x}) +A$, where $A$ is the adjacency matrix of the given network and $D(\lambda_{x})$ is a diagonal matrix of same size with diagonal values determine by a non negative function
 $\lambda_{x}$ of the node $x$.

The choice of $\lambda$ and $\lambda_{x}$ plays a crucial role in combination with the function chosen in the second module for determination of a suitable metric and is discussed later part of this subsection.

\subsubsection{Function selection} The function selection module determine the metric for pair of nodes. 
The function selector $\phi$ converts a similarity function (Pearson coefficient($\sigma_{PC}$), Spacerman coefficient($\sigma_{SC}$) or Cosine similarity($\sigma_{CS}$)) into a distance matrix. In general the similarity function satisfies the positivity and similarity condition of metric but not triangle inequality.
$\phi$ is a metric preserving ($\phi(d(x_i,x_j)= d_{\phi}(x_i,x_j)$), concave and monotonically increasing function. The three conditions above refer to as chord condition. The $\phi$ function is chosen to have minimum internal area with chord.

\subsubsection{Choice of $\lambda$ and $\phi(\sigma)()$} The choices in the above sub modules play a crucial role  in the graph to metric transformation algorithm to be used for community detection. The complex network is characterized by small average diameter and high clustering coefficient. Several studies on network structure analysis reveal that there are hub nodes and local nodes characterizing the interesting structure of the complex network.
 Suppose we have taken $\phi= arccos$, $\sigma_{CS}$ and constant $\lambda \geq 0$. $\lambda = 0$ penalize the effect of direct edge in the metric and is suitable to extract communities from highly dense graph. 
$\lambda = 1$ place the similar weight of direct edge and common neighbor reduce the effect of direct edge in the metric and is suitable to extract communities from moderately dense graph. $\lambda = 2$ set more importance to direct edge than common neighbor (this is the common case of available real networks). 
$\lambda \geq 2$ penalize the effect of common neighbor in the metric and is suitable to extract communities from very sparse graph. 
The choice of $\lambda$ depends on the input graph, i.e. whether it is sparse or dense and its cluster structure. A more detailed explanation on the metric described above can be obtained in \cite{SahaG15}.

\section{Nearest neighbor search on complex network using metric tree}
\label{nn-mtree}
There are a large number of methods developed to compute nearest neighbor search.
However, finding nearest neighbor search on some data where dimension is high suffer
from curse of dimensionality. Some recent research on this direction revealed that
dimension constrained can be tackled by using efficient data structures like metric tree
and locality sensitive hashing. In this section we have explored metric tree to perform
nearest  neighbor search on complex network with the help of metric mapping of complex network
described in the previous section.

\subsection{Metric-tree}

A metric tree is a data structure specially designed to perform nearest neighbor
query for the points residing on a metric space and perform well on high dimension
particularly when some approximation is permitted. 
A metric tree organizes a set of points in a spatial hierarchical manner. It is a
binary tree whose nodes represent a set of points. The root node represents all points, and
the points represented by an internal node $v$ is partitioned into two subsets, represented by
its two children. Formally, if we use $N(v)$ to denote the set of points represented by node
$v$, and use $v.lc$ and $v.rc$ to denote the left child and the right child of node $v$, then we have
$N(v) = N(v.lc) \cup N(v.rc)$
$\phi = N(v.lc) \cap N(v.rc)$
for all the non-leaf nodes. At the lowest level, each leaf node contains very few points.

An M-Tree \cite{CPZ97} has these components and sub-components:
\begin{itemize}

\item Non-leaf nodes: A set of routing objects $N_{RO}$, Pointer to Node's parent object $v_p$.
\item Leaf nodes: A set of objects $N_v$, Pointer to Node's parent object $v_p$.
\item Routing Object: (Feature value of) routing object $v_r$, Covering radius $r(v_r)$, Pointer to covering tree $T(v_r)$, Distance of $v_r$ from its parent object $d(v_r,P(v_r))$ 
\item Object: (Feature value of the) object $v_j$, Object identifier oid$(v_j)$,
 Distance of $v_j$ from its parent object $d(v_j,P(v_j))$

\end{itemize}

Partitioning: The key to building a metric-tree is how to partition a node v. 
A typical way is as follows: 
We first choose two pivot points from $N(v)$, denoted as $v.lpv$ and $v.rpv$.
Ideally, $v.lpv$ and $v.rpv$ are chosen so that the distance between them is the largest of 
all distances within $N(v)$. More specifically, 
$||v.lpv - v.rpv|| = max_{p_1,p_2\in N(v)} ||p_1 - p_2||$.
However, it takes $O(n^2 )$ time to find the optimal $v.lpv$ and $v.rpv$. In practice a
linear-time heuristic is used to find reasonable pivot points. $v.lpv$ and $v.rpv$
are then used to partition node $v$.
We first project all the points down to the
vector $u = v.rpv - v.lpv$, and then find the median point A along $u$. Next, we assign all
the points projected to the left of A to $v.lc$, and all the points projected to the right of $A$
to $v.rc$. We use $L$ to denote the $d-1$ dimensional plane orthogonal to $u$ and goes through
$A$. It is known as the decision boundary since all points to the left of $L$ belong to $v.lc$
and all points to the right of $L$ belong to $v.rc$. By using a median point to
split the data points, we can ensure that the depth of a metric-tree is $log n$. 
Each node $v$ also has a hypersphere $B$, such that all points represented by $v$ fall in the ball
centered at $v.center$ with radius $v.r$, i.e. $N(v) \in B(v.center, v.r)$.

Searching: A search on a metric-tree is performed using a stack. The current radius $r$ is used to decide which child node to search first. If
the query $q$ is on the left of current point, then $v.lc$ is searched first, otherwise, $v.rc$ is searched first. At all times, the algorithm maintains a candidate NN and there distance determine the current radius, which is the nearest neighbor it finds so far while traversing the tree. 
We call this point $x$, and denote the distance between $q$ and $x$ by $r$. If algorithm is about to exploit a node $v$, but discovers that no member of $v$ can be within
distance $r$ of $q$, then it skip the subtree from $v$. This happens whenever $v.center - |q - v.r| \geq r$. In practice, metric tree search typically finds a very good NN
candidate quickly, and then spends lots of the time verifying that it is in fact the true
NN. However, in case of approximate NN we can save majority of time with moderate approximation guarantee. 
The algorithm for NN search using metric tree is given below \ref{alg-nn-m}.

\subsection{Nearest Neighbor search algorithm using M-Tree}
algorithm
\begin{algorithm}[H]
\begin{algorithmic}[1]
\REQUIRE $M= (V, d)$ $\&$ $q$
\ENSURE $d(q, v_q)$
\STATE Insert root object $v_r$ in stack
\STATE Set current radius as $d(v_r, q)$  
\STATE Successively traverse the tree in search of $q$ 
\STATE PUSH all the objects of traversal path into stack 
\STATE Update the current radius
\STATE If leaf object reached
\STATE POP objects from stack
\STATE For all points lying inside the ball of current radius centering $q$, verify for possible nearest neighbor and update the current radius.
\RETURN $d(q, v_q)$
\end{algorithmic}
\caption{NN search in M-Tree}
\label{alg-nn-m}
\end{algorithm}

\begin{Theorem}
Let  $M=(V,d)$, be a  bounded metric space. Then for any fixed data $V\in R^n$ of size $n$, and for constant
$c \geq 1$, $\exists \epsilon$   such that we may compute $d(q, V)|_\epsilon$   
with at most $c \cdot \left \lceil \log(n)+1 \right \rceil$  expected metric evaluations\cite{CPZ98b}
\end{Theorem}

\section{Nearest neighbor search on complex network using locality sensitive hashing}
\label{nn-approx}

Metric trees, so far represent the practical state of the art for
achieving efficiency in the largest dimensionality possible. However, many real-world 
problems are posed with very large dimensionality which are beyond the capability
of such search structures to achieve sub-linear efficiency. Thus, the high-dimensional case
is the long-standing frontier of the nearest-neighbor problem.

The approximate nearest neighbor can be computed very efficiently using Locality sensitive hashing.

\subsection{Approximate nearest neighbor}

Given a metric space $(S, d)$ and some finite subset $S_D$ of data points$S_D \subset S$ on which 
nearest neighbor queries are to be made, our aim to organize $S_D$ s.t. NN queries can be answered more efficiently.
 For any $q\in S$, NN problem consists of finding single minimal located point $p\in S_D$ s.t. $d(p,q)$ is minimum over all $p\in S_D$. We denote this by $p= NN(q, S_D)$.

An $\epsilon $ approximate NN of $q\in S$ is to find a point $p\in S_D$ s.t. $d(p,q) \leq (1+\epsilon) d(x,d)$ $\forall $ $x\in S_D$.

\subsection{Locality Sensitive Hashing (LSH)}
Several methods to compute first nearest neighbor query exists in the
literatures and locality-sensitive hashing (LSH) is most popular
because of its dimension independent runtime \cite{MotwaniNP07, Pauleve:2010}. 
In a locality sensitive hashing, the hash function has the property that close points are hash into 
same bucket with high probability and distance points are hash into same bucket with low
probability. 
Mathematically, a family $H = \{ h : S \rightarrow U \}$ is called $(r_1, r_2, p_1, p_2)$-sensitive 
if for any $p, q \in S$
\begin{itemize}
\item if $p \in B(q,r_1)$ then $ Pr_H[h(q)=h(p)] \geq p_1 $
\item if $p \notin B(q,r_2)$ then $ Pr_H[h(q)=h(p)] \leq p_2 $
\end{itemize}
where $B(q, r)$ denotes a hypersphere of radius $r$ centered at $q$. 
In order for a locality-sensitive family to be useful, it has to satisfy inequalities 
$p_1 > p_2$ and $r_1 < r_2$
when $D$ is a dissimilarity measure, or $p_1 > p_2$ and $r_1 > r_2$ when $D$ is a 
similarity measure\cite{Indyk98,Gionis}. The value of $\delta = log(1/P_1 )/log(1/P_2 )$ determines search performance of LSH. Defining a LSH as $a (r, r(1 + \epsilon), p1 , p2 )$, the $(1 + \epsilon)$ NN problem can be solved via series of hashing and searching within the buckets \cite{JolyB08, Datar:2004, Gionis:1999}.

\subsection{Locality sensitive hash function for complex network}
In this sub-section, we discuss the existence of locality sensitive hash function families for
the proposed metric on complex network. The LSH data structure stores all nodes in hash tables and
searches for nearest neighbor via retrieval. The hash
table is contain many buckets and identified by bucket id. 
Unlike conventional hashing the LSH approach try to maximize
the probability of collision of near items and put them into same bucket. 
For any given the query $q$ the bucket $h(q)$  considered to search nearest node.
In general $k$ hash functions are chosen independently and
uniformly at random from hash family $H$. The output of nearest neighbor
query is provided from the union ok $k$ buckets. The consensus of $k$ functions
reduces the error of approximation.
For metric defined in the previous section \ref{nn-near} we considered
$k$ random points from the metric space. Each random point $r_i$ define a hash function
$h_{i}(x) = sign(d(x, r_i))$, where $d$ is the metric and $i \in [1,k]$. These randomized hash functions
are locality sensitive \cite{Andoni:2008, Charikar:2002}.

\begin{algorithm}[H]
\begin{algorithmic}[1]
\REQUIRE $M= (V, d)$ $\&$ $q$
\ENSURE $d(q, V)$
\STATE Identify buckets of query point $q$ corresponding to different hash functions.
\STATE Compute nearest neighbor of $q$ only for the points inside the selected buckets.
\RETURN $d(q, V)$
\end{algorithmic}
\caption{NN search in LSH}
\label{alg4}
\end{algorithm}

\begin{Theorem}
Let  $M=(V,d)$, be a  bounded metric space. Then for any fixed data $V\in R^n$ of size $n$, and for constant
$c \geq 1$, $\exists \epsilon$   such that we may compute $d(q, V)|_\epsilon$   
with at most $mn^{O(1/\epsilon )}$  expected metric evaluations, where $m$ is the number of dimension of the metric space. In case of complex network $m=n$ so expected time is $n^{O(2/\epsilon)}$
\cite{CPZ98b, Indyk_asublinear}.
\end{Theorem}

\section{Network community detection using nearest neighbor search}
\label{nn-comm}
Community detection in real networks aims to capture the structural organization of the network using the connectivity information as input\cite{NewGir04,Schaeffer200727}.
Early work on this domain was attempted by Weiss and
Jacobson while searching for a work group within a government agency\cite{weiss}.

Most of the methods developed for network community detection are based on a two-step approach.
The first step is specifying a quality measure (evaluation measure, objective function) that quantifies the desired properties of communities and the second step is applying an algorithmic techniques to assign the nodes of graph into communities by optimizing the objective function.

Several measures for quantifying the quality of communities have been proposed, they mostly consider that communities are set of nodes with many edges between them and few connections with nodes of different communities(e.g. modularity, conductance, expansion, internal density, average degree, triangle precipitation ratio,..,e.t.c. ). 

\subsection{Popular algorithms}
In this subsection we have given a brief list of the algorithms developed for network community detection purposes. The broad categorization of the algorithms are based on graph traversal, semidefinite programming and spectral. The basic approach and the complexity of very popular algorithms are listed in the table \ref{cnalgo}.
There are more algorithms developed to solve network community detection problem a complete list can be obtained in several survey articles \cite{Fortunato10, LeskovecLM10, YangL12}. 

A partial list of algorithms developed for network community detection purpose is tabulated in \ref{cnalgo}. The algorithms are categorized into three main group as spectral (SP), graph traversal based (GT) and semi-definite programming based (SDP). The categories and complexities are also given in the table \ref{cnalgo}.   
\begin{table}[h!]
\caption{Algorithms for network community detection and their complexities}
\label{cnalgo}
\centering
\tiny
\begin{tabular}{|l || c ||l ||l |} 
 \hline
Author  & Ref.  & Cat. & Order  \\[0.5ex]\hline\hline
Van Dongen 	& (Graph clustering, 2000\cite{dongen})  & GT  & $O(nk^2 )$, $k < n$ parameter \\\hline
Eckmann \& Moses  &  (Curvature, 2002\cite{Eckmann2002})  &  GT  & $O(m k^2 )$  \\\hline
Girvan \& Newman  &  (Modularity, 2002\cite{girvan02})  &  SDP  &  $O(n^2 m)$ \\\hline
Zhou \& Lipowsky  & (Vertex Proximity, 2004\cite{ZhouL04}) & GT  &  $O(n^3 )$  \\\hline
Reichardt \& Bornholdt  & (spinglass, 2004\cite{spinglass})  &  SDP  & parameter dependent \\\hline
Clauset et al.  	& (fast greedy, 2004\cite{Clauset2004}) & SDP  & $O(n log_2 n)$ \\\hline
Newman \& Girvan  	& (eigenvector, 2004 \cite{NewGir04})  &  SP  &  $O(nm^2 )$ \\\hline
Wu \& Huberman  &  (linear time, 2004\cite{linear})  &  GT  &  $O(n + m)$\\\hline
Fortunato et al.  &  (infocentrality, 2004\cite{infocentrality})  &  SDP & $O(m^3 n)$ \\\hline
Radicchi et al.  &  (Radicchi et al., 2004\cite{radicchi2004})  &  SP  &  $O(m^4 /n^2 )$\\\hline
Donetti \& Munoz  &  (Donetti and Munoz, 2004\cite{Donetti})  &  SDP  &  $O(n^3 )$ \\\hline
Guimera et al. &  (Simulated Annealing, 2004\cite{Guimera04})  &  SDP  &  parameter dependent \\\hline 
Capocci et al.  &  (Capocci et al., 2004\cite{Capocci04})  &  SP  &  $O(n^2 )$ \\\hline
Latapy \& Pons 		& (walktrap, 2004\cite{Pons04}) & SP  & $O(n^3 )$  \\\hline
Duch \& Arenas  &  (Extremal Optimization, 2005\cite{duch05})  &  GT  &  $O(n^2 log n)$ \\\hline
Bagrow \& Bollt   &  (Local method, 2005\cite{Bagrow05})  &  SDP  &  $O(n^3 )$ \\\hline
Palla et al.  &  (overlapping community, 2005\cite{palla05})  &  GT  &  $O(exp(n))$ \\\hline
Raghavan et al.  &  (label propagation, 2007\cite{raghavan})  &  GT  &  $O(n + m)$\\\hline
Rosvall \& Bergstrom  & (Infomap, 2008\cite{rosvall2008random})  &  SP  & $O(m)$ \\\hline
Ronhovde \& Nussinov  &  (Multiresolution community, 2009\cite{Ronhovde09})  &  GT  &  $O(m\beta log n)$, $\beta \approx 1.3$ \\[1ex]\hline
\end{tabular}
\end{table}

\subsection{k-central algorithm for network community detection using nearest neighbor search}
\label{gtm-comm}
In this section we have described k-central algorithm for the purpose of network community detection by using the nearest neighbor search inside complex network. We have also studied and analyzed the advantages of the k-central method over the standard algorithm for network community detection.

\subsubsection{k-central algorithm}
\label{partition}
The community detection methods based on partitioning of graph is possible using nearest neighbor search, because the nodes of the graph are converted into the points of a metric space. This algorithm for network community detection converges automatically and does not compute the value of objective function in iterations therefore reduce the computation compared to standard methods. The results of this algorithm are competitive on a large set networks shown in section \ref{nn-res}. The k-central algorithm for community detection and its details analysis is given below. 

\subsubsection{k selection}
\label{k-selection} Determining the optimal number of k is an important problem for community detection researchers. An extensive analysis can be found in the work of Leskovec et al.  \cite{Leskovec2008}. The standard practice is to solve an optimization equation with respect to k for which the optimal value of the objective function is achieved. Another method based on farthest first traversal is also very useful in terms of computational efficiency \cite{Gonzalez85}. For small networks the global optimization works better and for very large network the second choice give the faster approximate solution.

\subsubsection{Initialization} The set of initial nodes are also very important problem for k-central algorithm
\begin{itemize}
\item Input: graph $G = (V,E)$, with the node similarity $sim(x_a,x_b)$ defined on it
\item Output: A partition of the nodes into $k$ communities $C_1, C_2,..., C_k $
\item Objective function: Maximize the minimum intra community similarity 
$$min_{j\in\{1,2,..,k\}}max_{x_a,x_b \in C_j} ~~ sim(x_a,x_b)$$ 

\end{itemize}

\begin{algorithm}[H]
\begin{algorithmic}[1]
\REQUIRE $M= (V, d)$ 
\ENSURE $T= \{C_1, C_2, \dots, C_k \}$ with minimum $cost(T)$
\STATE  Initialize centers $z_1, \dots, z_k \in R^n$ and clusters $T= \{C_1, C_2, \dots, C_k\}$ 
\REPEAT 
	\FOR{$i=1$ to $k$}	
		\FOR{$j=1$ to $k$}
			\STATE $C_i \leftarrow \{x \in V ~s.t.~ |z_i-x| \leq |z_j-x|\}$
		\ENDFOR
	\ENDFOR 
	\FOR{$j=1$ to $k$}	
		\STATE $z_i \leftarrow Central(C_i)$ ; where $Central(C_i)$ returns the node with minimum total distance in the class of consideration.
	\ENDFOR 
\UNTIL $|cost(T_t) - cost(T_{t+1})| = 0$
\RETURN $T= \{C_1, C_2,\dots, C_k \}$

\end{algorithmic}
\caption{k-central algorithm}
\label{alg2}
\end{algorithm}

\subsubsection{Convergence}
\label{convergence} Convergence of the network community detection algorithms are the least studied research areas of network science. However, the rate of convergence is one of the important issues and low rate of convergence is the major pitfall of the most of the existing algorithms. Due to the transformation into the metric space, our algorithm equipped with the quick convergence facility of the k-partitioning on metric space by providing a good set of initial points. Another crucial pitfall suffer by majority of the existing algorithms is the validation of the objective function in each iteration during convergence. Our algorithm converges automatically to the optimal partition thus reduces the cost of validation during convergence. 

\begin{Theorem}
During the course of the $k$ center partitioning algorithm, the cost monotonically decreases.
\end{Theorem}

\begin{proof}
Let $Z^t = \{z_1^t ,\dots , z_k^t\}$ , $T^t = \{ C_1^t ,\dots, C_k^t\}$ denote the centers and clusters at the start of the  $t^{th}$ iteration of $k$ partitioning algorithm. The first step of the iteration assigns each data point to its closest center; therefore $cost(T^{t+1},Z^t) \leq cost(T^t,Z^t) $

On the second step, each cluster is re-centered at its mean; therefore $cost(T^{t+1},Z^{t+1}) \leq cost(T^{t+1},Z^t)$ 

\end{proof}

\section{Experiments and results}
\label{nn-res}
We performed many of experiments to test the performance of nearest neighbor search based community detection method for complex network over several real networks\ref{cndata}. Objective of the experiment is to verify behavior of the algorithm and the time required to compute the algorithm. One of the major goals of the experiment is to verify the behavior of the algorithm with respect to the performance of other popular methods
exists in the literature with respect to the standard measures like conductance and modularity.
Experiments are conducted to compare the results (tables \ref{nnres}, \ref{nnresmod} and \ref{nnrestime}) of our algorithm with the state of the art algorithms (table \ref{cnalgo}) available in the literature in terms of common measures mostly used by the researchers of the domain of network community detection. The details of the several experiments and the analysis of the results are given in the following subsections.

\subsection{Experimental designs}
Experiment for comparison: In this experiment we have compared several algorithms for network community detection with our proposed algorithm developed using nearest neighbor search in complex network. Experiment is performed on a large list of network data sets. Two version of the experiment is developed for comparison purpose based on two different quality measure conductance and modularity. The results are shown in the tables \ref{nnres} and \ref{nnresmod} respectively.

Experiment on performance and time: In this experiment we have evaluated our algorithm for performance on the network collection\ref{cndata}. We have evaluated the time taken by our algorithm on different size of networks and is shown in the table \ref{nnrestime}.


\subsection{Performance indicator}
Modularity: The notion of modularity is the most popular for the network community detection purpose. 
The modularity index assigns high scores to communities whose internal edges are more than that expected in a random-network model which preserves the degree distribution of the given network.

Conductance: Conductance is widely used for graph partitioning literature. 
The conductance of a set $S$ with complement $S^{C}$ is the ratio of the number of edges connecting nodes in $S$ to nodes in $S^{C}$ by the total number of edges incident to $S$ or to $S^{C}$ (whichever number is smaller).

\subsection{Datasets}

A list of real networks taken from several real life interactions is considered for our experiments and they are tabulate \ref{cndata} below. We have also listed the number of nodes, number of edges, average diameter, data complexity for community detection (DCC) and the k value used (\ref{k-selection}). The values of the last column can be used to assess the quality of detected communities.

\begin{table}[!htbp]
\caption{Complex network datasets and values of their parameters}
\label{cndata}
\centering
\tiny
\begin{tabular}{|l || c l l ||c | c|} 
 \hline
 Name & Type  & \# Nodes & \# Edges & Diameter & k \\ [0.5ex] 
 \hline\hline
DBLP	 	& U & 317,080	& 1,049,866 	& 8   &  268 \\\hline	
Arxiv-AstroPh	& U & 18,772	& 396,160 	& 5   &  23 \\\hline	
web-Stanford	& D & 281,903 	& 2,312,497 	& 9.7 &  69 \\\hline	
Facebook 	& U & 4,039 	& 88,234	& 4.7 &  164 \\\hline	
Gplus 		& D & 107,614 	& 13,673,453 	& 3   &  457 \\\hline	
Twitter		& D & 81,306 	& 1,768,149 	& 4.5 &  213 \\\hline	
Epinions1 	& D & 75,879 	& 508,837 	& 5   &  128  \\\hline	
LiveJournal1 	& D & 4,847,571 & 68,993,773 	& 6.5 &  117 \\\hline	
Orkut		& U & 3,072,441	& 117,185,083 	& 4.8 &  756  \\\hline	
Youtube		& U & 1,134,890 & 2,987,624 	& 6.5 &  811 \\\hline	
Pokec 		& D & 1,632,803 & 30,622,564 	& 5.2 &  246 \\\hline	
Slashdot0811 	& D & 77,360 	& 905,468 	& 4.7 &  81  \\\hline	
Slashdot0922 	& D & 82,168 	& 948,464 	& 4.7 &  87  \\\hline	
Friendster	& U & 65,608,366& 1,806,067,135 & 5.8 &  833 \\\hline	
Amazon0601	& D & 403,394	& 3,387,388	& 7.6 &  92 \\\hline	
P2P-Gnutella31	& D & 62,586	& 147,892 	& 6.5 &  35 \\\hline	
RoadNet-CA	& U & 1,965,206 & 5,533,214 	& 500 &  322 \\\hline	
Wiki-Vote 	& D & 7,115 	& 103,689 	& 3.8 &  21 \\ [1ex] 	
 \hline 
\end{tabular}
\end{table}

\subsection{Computational results}
In this subsection we have compared two groups of algorithms for network community detection with our proposed algorithm using nearest neighbor search. Experiment is performed on a large list of network data sets. Two version of the experiment is developed for comparison purpose based on two different quality measure conductance and modularity. The results based on conductance is shown in the table \ref{nnres}  and the results based on modularity is shown in the table \ref{nnresmod}, respectively. Regarding the two groups of algorithms; first group contain algorithms based on semi-definite programming and the second group contain algorithms based on graph traversal approaches. For each group, we have taken the best value of conductance in table \ref{nnres} and best value of modularity in table \ref{nnresmod} among all the algorithms in the groups. The results obtained with our approach are very competitive with most of the well known algorithms in the literature and this is justified over the large collection of datasets. On the other hand, it can be observed that time taken (table \ref{nnrestime}) by our algorithm is quite less compared to other methods and justify the theoretical findings. 

\begin{table}[h!]
\caption{Comparison of our approaches with other best methods in terms of conductance}
\label{nnres}
\centering
\tiny
\begin{tabular}{|l || c c c ||c c c |} 
 \hline
 Name 		& Spectral  & SDP & GT & Index & M-tree & LSH 			 \\ [0.5ex] 
 \hline\hline
Facebook   	&    0.0097 &  0.1074 &  0.1044 &  0.1082 &  0.0827 &  0.0340  \\\hline
Gplus    	&    0.0119 &  0.1593 &  0.1544 &  0.1602 &  0.1207 &  0.0500  \\\hline
Twitter   	&    0.0035 &  0.0480 &  0.0465 &  0.0483 &  0.0363 &  0.0150  \\\hline
Epinions1   	&    0.0087 &  0.1247 &  0.1208 &  0.1254 &  0.0941 &  0.0390  \\\hline
LiveJournal1   	&    0.0039 &  0.0703 &  0.0680 &  0.0706 &  0.0523 &  0.0218  \\\hline
Pokec    	&    0.0009 &  0.0174 &  0.0168 &  0.0175 &  0.0129 &  0.0054  \\\hline
Slashdot0811   	&    0.0005 &  0.0097 &  0.0094 &  0.0098 &  0.0072 &  0.0030  \\\hline
Slashdot0922   	&    0.0007 &  0.0138 &  0.0133 &  0.0138 &  0.0102 &  0.0043  \\\hline
Friendster  	&    0.0012 &  0.0273 &  0.0263 &  0.0273 &  0.0200 &  0.0084  \\\hline
Orkut   	&    0.0016 &  0.0411 &  0.0397 &  0.0412 &  0.0300 &  0.0126  \\\hline
Youtube   	&    0.0031 &  0.0869 &  0.0838 &  0.0871 &  0.0633 &  0.0267  \\\hline
DBLP  		&    0.0007 &  0.0210 &  0.0203 &  0.0211 &  0.0152 &  0.0064  \\\hline
Arxiv-AstroPh  	&    0.0024 &  0.0929 &  0.0895 &  0.0931 &  0.0669 &  0.0283  \\\hline
web-Stanford  	&    0.0007 &  0.0320 &  0.0308 &  0.0320 &  0.0229 &  0.0097  \\\hline
Amazon0601  	&    0.0018 &  0.0899 &  0.0865 &  0.0900 &  0.0643 &  0.0273  \\\hline
P2P-Gnutella31  &    0.0009 &  0.0522 &  0.0503 &  0.0523 &  0.0373 &  0.0158  \\\hline
RoadNet-CA  	&    0.0024 &  0.1502 &  0.1445 &  0.1504 &  0.1070 &  0.0455  \\\hline
Wiki-Vote   	&    0.0026 &  0.1853 &  0.1783 &  0.1855 &  0.1318 &  0.0561  \\\hline
 \hline
\end{tabular}
\end{table}

\begin{table}[h!]
\caption{Comparison of our approaches with other best methods in terms of modularity}
\label{nnresmod}
\centering
\tiny
\begin{tabular}{|l || c c c ||c c c |} 
 \hline
 Name 		& Spectral  & SDP & GT & Index & M-tree & LSH 			 \\ [0.5ex] 
 \hline\hline
Facebook   	& 0.4487 & 0.5464 & 0.5434 & 0.5472 & 0.5450 &	0.5421  \\\hline
Gplus   	& 0.2573 & 0.4047 & 0.3998 & 0.4056 & 0.4041 &	0.4021  \\\hline
Twitter   	& 0.3261 & 0.3706 & 0.3691 & 0.3709 & 0.3692 &	0.3669  \\\hline
Epinions1  	& 0.0280 & 0.1440 & 0.1401 & 0.1447 & 0.1443 &	0.1437 \\\hline
LiveJournal1  	& 0.0791 & 0.1455 & 0.1432 & 0.1458 & 0.1450 &	0.1439 \\\hline
Pokec   	& 0.0129 & 0.0294 & 0.0288 & 0.0295 & 0.0292 &	0.0287  \\\hline
Slashdot0811   	& 0.0038 & 0.0130 & 0.0127 & 0.0131 & 0.0129 &	0.0127  \\\hline
Slashdot0922   	& 0.0045 & 0.0176 & 0.0171 & 0.0176 & 0.0174 &	0.0172  \\\hline
Friendster  	& 0.0275 & 0.0536 & 0.0526 & 0.0536 & 0.0531 &	0.0525 \\\hline
Orkut   	& 0.0294 & 0.0689 & 0.0675 & 0.0690 & 0.0685 &	0.0678  \\\hline
Youtube   	& 0.0096 & 0.0934 & 0.0903 & 0.0936 & 0.0934 &	0.0930  \\\hline
DBLP 		& 0.4011 & 0.4214 & 0.4207 & 0.4215 & 0.4196 &	0.4171  \\\hline
Arxiv-AstroPh   & 0.4174 & 0.5079 & 0.5045 & 0.5081 & 0.5061 &	0.5035 \\\hline
web-Stanford   	& 0.3595 & 0.3908 & 0.3896 & 0.3908 & 0.3890 &	0.3866  \\\hline
Amazon0601 	& 0.1768 & 0.2649 & 0.2615 & 0.2650 & 0.2637 &	0.2621  \\\hline
P2P-Gnutella31  & 0.0009 & 0.0522 & 0.0503 & 0.0523 & 0.0523 &	0.0523 \\\hline
RoadNet-CA 	& 0.0212 & 0.1690 & 0.1633 & 0.1692 & 0.1680 &	0.1664 \\\hline
Wiki-Vote 	& 0.0266 & 0.2093 & 0.2023 & 0.2095 & 0.2090 &	0.2083  \\[1ex] 
 \hline
\end{tabular}
\end{table}

\begin{table}[!htbp]
\caption{Comparison of our approaches with other best methods in terms of time}
\label{nnrestime}
\centering
\tiny
\begin{tabular}{|l || c c c ||c c c |} 
 \hline
 Name 		& Spectral  & SDP & GT & Index & M-tree & LSH 			 \\ [0.5ex] 
 \hline\hline
Facebook   	&  6 &  7 &  11 &  6 &  4 &  1  \\\hline
Gplus   	&  797 &  832 &  1342 &  661 &  390 &  115  \\\hline
Twitter   	&  462 &  485 &  786 &  398 &  235 &  68  \\\hline
Epinions1  	&  411 &  419 &  667 &  292 &  174 &  56 \\\hline
LiveJournal1  	&  1297 &  1332 &  2129 &  969 &  576 &  179 \\\hline
Pokec   	&  1281 &  1305 &  2075 &  901 &  538 &  173  \\\hline
Slashdot0811   	&  552 &  561 &  891 &  382 &  228 &  74  \\\hline
Slashdot0922   	&  561 &  570 &  906 &  389 &  232 &  75  \\\hline
Friendster  	&  2061 &  2105 &  3352 &  1477 &  880 &  280 \\\hline
Orkut   	&  1497 &  1529 &  2435 &  1074 &  640 &  203  \\\hline
Youtube   	&  829 &  844 &  1340 &  578 &  345 &  111  \\\hline
DBLP 		&  381 &  403 &  655 &  341 &  201 &  57  \\\hline
Arxiv-AstroPh   &  217 &  230 &  375 &  197 &  116 &  33 \\\hline
web-Stanford   	&  498 &  525 &  852 &  437 &  258 &  74  \\\hline
Amazon0601 	&  653 &  678 &  1089 &  520 &  308 &  93  \\\hline
P2P-Gnutella31  &  182 &  184 &  293 &  124 &  74 &  24 \\\hline
RoadNet-CA 	&  758 &  785 &  1261 &  599 &  355 &  107 \\\hline
Wiki-Vote 	&  54 &  55 &  88 &  39 &  23 &  7  \\[1ex] 
 \hline
\end{tabular}
\end{table}

\subsection{Results analysis and achievements}
In this subsection, we have described the analysis of the results obtained in our experiments shown above and also highlighted the achievements from the results. It is clearly evident from the results shown in the tables \ref{nnres}, \ref{nnresmod} and \ref{nnrestime} that, proposed nearest neighbor based method for network community detection using metric tree and locality sensitive hashing provide very good competitive performance with respect to conductance and modularity and also in terms of time. It is also evident from the results that our methods provide case base solution of network community detection depending on the requirements of time or better conductance/modularity.


\section{Conclusions}
In this paper, we have studied the interesting problem of
nearest neighbor queries in complex networks.
Processing nearest neighbor search in complex networks cannot be
achieved by straightforward applications of previous
approaches for the Euclidean space due to the
complexity of graph traversal based computations of node nearness as opposed to
geometric distances. We presented the transformation of graph to metric space
and efficient computation of nearest neighbor therein using metric tree and locality sensitive hashing. 
Our techniques can be applied for various structural analysis of complex network using geometric approaches. 
To validate the performance of proposed nearest neighbor search designed for complex networks we applied
our approaches on community detection problem. The results obtained on several network data sets prove the usefulness of the proposed method and provide motivation for further application of other structural analysis
of complex network using nearest neighbor search.


\acknowledgments{Acknowledgments}

This work is supported by the Jaypee University of Information Technology.


\authorcontributions{Author Contributions}

Suman Saha proposed the algorithm and prepared the manuscript. S.P. Ghrera was in charge of the
overall research and critical revision of the paper.


\conflictofinterests{Conflicts of Interest}

The authors declare no conflict of interest. 

\bibliographystyle{mdpi}
\makeatletter
\renewcommand\@biblabel[1]{#1. }
\makeatother



\begin{thebibliography}{-------}
\providecommand{\natexlab}[1]{#1}

\bibitem[Uhlmann(1991)]{UHLMANN1991}
Uhlmann, J.K.
\newblock Satisfying general proximity / similarity queries with metric trees.
\newblock {\em Information Processing Letters} {\bf 1991}, {\em 40},~175 --
  179.

\bibitem[Ruiz(1986)]{VIDALRUIZ1986}
Ruiz, E.V.
\newblock An algorithm for finding nearest neighbours in (approximately)
  constant average time.
\newblock {\em Pattern Recognition Letters} {\bf 1986}, {\em 4},~145 -- 157.

\bibitem[Panigrahy(2006)]{Panigrahy:2006}
Panigrahy, R.
\newblock Entropy Based Nearest Neighbor Search in High Dimensions.
\newblock  Proceedings of the Seventeenth Annual ACM-SIAM Symposium on Discrete
  Algorithm; Society for Industrial and Applied Mathematics: Philadelphia, PA,
  USA,  2006; SODA '06, pp. 1186--1195.

\bibitem[Indyk and Motwani(1998)]{Indyk:1998}
Indyk, P.; Motwani, R.
\newblock Approximate Nearest Neighbors: Towards Removing the Curse of
  Dimensionality.
\newblock  Proceedings of the Thirtieth Annual ACM Symposium on Theory of
  Computing; ACM: New York, NY, USA,  1998; STOC '98, pp. 604--613.

\bibitem[Gionis \em{et~al.}(1999)Gionis, Indyk, and Motwani]{Gionis:1999}
Gionis, A.; Indyk, P.; Motwani, R.
\newblock Similarity Search in High Dimensions via Hashing.
\newblock  Proceedings of the 25th International Conference on Very Large Data
  Bases; Morgan Kaufmann Publishers Inc.: San Francisco, CA, USA,  1999; VLDB
  '99, pp. 518--529.

\bibitem[Dasgupta and Freund(2008)]{Dasgupta:2008}
Dasgupta, S.; Freund, Y.
\newblock Random Projection Trees and Low Dimensional Manifolds.
\newblock  Proceedings of the Fortieth Annual ACM Symposium on Theory of
  Computing; ACM: New York, NY, USA,  2008; STOC '08, pp. 537--546.

\bibitem[Akoglu \em{et~al.}(2014)Akoglu, Khandekar, Kumar, Parthasarathy,
  Rajan, and Wu]{Akoglu:2014}
Akoglu, L.; Khandekar, R.; Kumar, V.; Parthasarathy, S.; Rajan, D.; Wu, K.L.
\newblock Fast Nearest Neighbor Search on Large Time-Evolving Graphs.
\newblock  Proceedings of the European Conference on Machine Learning and
  Knowledge Discovery in Databases - Volume 8724; Springer-Verlag New York,
  Inc.: New York, NY, USA,  2014; ECML PKDD 2014, pp. 17--33.

\bibitem[Liu \em{et~al.}(2004)Liu, Moore, Gray, and Yang]{Liu04}
Liu, T.; Moore, A.W.; Gray, E.; Yang, K.
\newblock An investigation of practical approximate nearest neighbor
  algorithms.
\newblock  NIPS2004. MIT Press,  2004, pp. 825--832.

\bibitem[Freeman(1978)]{Freeman78centralityin}
Freeman, L.C.
\newblock Centrality in social networks conceptual clarification.
\newblock {\em Social Networks} {\bf 1978}, p. 215.

\bibitem[Carrington \em{et~al.}(2005)Carrington, Scott, and
  Wasserman]{CarringtonScott:2005}
Carrington, P.J.; Scott, J.; Wasserman, S., Eds.
\newblock {\em Models and methods in social network analysis}; Cambridge
  University Press,  2005.

\bibitem[Newman(2003)]{newman2003structure}
Newman, M.
\newblock The Structure and Function of Complex Networks.
\newblock {\em SIAM review} {\bf 2003}, {\em 45},~167--256.

\bibitem[Radicchi \em{et~al.}(2004)Radicchi, Castellano, Cecconi, Loreto, and
  Parisi]{radicchi2004}
Radicchi, F.; Castellano, C.; Cecconi, F.; Loreto, V.; Parisi, D.
\newblock {Defining and identifying communities in networks}.
\newblock {\em Proceedings of the National Academy of Sciences} {\bf 2004},
  {\em 101},~2658.

\bibitem[Fortunato(2010)]{Fortunato10}
Fortunato, S.
\newblock Community detection in graphs.
\newblock {\em Physics Reports} {\bf 2010}, {\em 486},~75 -- 174.

\bibitem[Weiss and Jacobson(1955)]{weiss}
Weiss, R.; Jacobson, E.
\newblock A Method for the Analysis of Complex Organisations.
\newblock {\em American Sociological Review} {\bf 1955}, {\em 20},~661--668.

\bibitem[Schaeffer(2007)]{Schaeffer200727}
Schaeffer, S.E.
\newblock Graph clustering.
\newblock {\em Computer Science Review} {\bf 2007}, {\em 1},~27 -- 64.

\bibitem[Newman and Girvan(2004)]{NewGir04}
Newman, M.E.J.; Girvan, M.
\newblock Finding and evaluating community structure in networks.
\newblock {\em Physical Review} {\bf 2004}, {\em E 69}.

\bibitem[Luxburg(2007)]{Luxburg07}
Luxburg, U.
\newblock A tutorial on spectral clustering.
\newblock {\em Statistics and Computing} {\bf 2007}, {\em 17},~395--416.

\bibitem[Pons and Latapy(2004)]{Pons04}
Pons, P.; Latapy, M.
\newblock Computing communities in large networks using random walks.
\newblock {\em J. of Graph Alg. and App.} {\bf 2004}, {\em 10},~284--293.

\bibitem[Duch and Arenas(2005)]{duch-2005}
Duch, J.; Arenas, A.
\newblock Community detection in complex networks using Extremal Optimization.
\newblock {\em Physical Review E} {\bf 2005}, {\em 72},~027104.

\bibitem[Chakrabarti(2004)]{Chakrabarti04a}
Chakrabarti, D.
\newblock AutoPart: Parameter-Free Graph Partitioning and Outlier Detection.
\newblock  PKDD; Boulicaut, J.F.; Esposito, F.; Giannotti, F.; Pedreschi, D.,
  Eds. Springer,  2004, Vol. 3202, {\em Lecture Notes in Computer Science}, pp.
  112--124.

\bibitem[Macropol and Singh(2010)]{MacropolS10}
Macropol, K.; Singh, A.K.
\newblock Scalable Discovery of Best Clusters on Large Graphs.
\newblock {\em PVLDB} {\bf 2010}, {\em 3},~693--702.

\bibitem[Levorato and Petermann(2011)]{LevoratoP11}
Levorato, V.; Petermann, C.
\newblock Detection of communities in directed networks based on strongly
  p-connected components.
\newblock  CASoN. IEEE,  2011, pp. 211--216.

\bibitem[Brandes \em{et~al.}(2003)Brandes, Gaertler, and Wagner]{BrandesGW03}
Brandes, U.; Gaertler, M.; Wagner, D.
\newblock Experiments on Graph Clustering Algorithms.
\newblock  ESA; Battista, G.D.; Zwick, U., Eds. Springer,  2003, Vol. 2832,
  {\em Lecture Notes in Computer Science}, pp. 568--579.

\bibitem[Bullmore and Sporns(2009)]{bullmore2009complex}
Bullmore, E.; Sporns, O.
\newblock {Complex brain networks: graph theoretical analysis of structural and
  functional systems}.
\newblock {\em Nature Reviews Neuroscience} {\bf 2009}, {\em 10},~186--198.

\bibitem[Saha and Ghrera(2015)]{SahaG15}
Saha, S.; Ghrera, S.P.
\newblock Network Community Detection on Metric Space.
\newblock {\em Algorithms} {\bf 2015}, {\em 8},~680--696.

\bibitem[Ciaccia \em{et~al.}(1997)Ciaccia, Patella, and Zezula]{CPZ97}
Ciaccia, P.; Patella, M.; Zezula, P.
\newblock {M}-tree: An Efficient Access Method for Similarity Search in Metric
  Spaces.
\newblock  Proceedings of the 23rd International Conference on Very Large Data
  Bases ({VLDB'97}); Morgan Kaufmann Publishers, Inc.: Athens, Greece,  1997;
  pp. 426--435.

\bibitem[Ciaccia \em{et~al.}(1998)Ciaccia, Patella, and Zezula]{CPZ98b}
Ciaccia, P.; Patella, M.; Zezula, P.
\newblock A Cost Model for Similarity Queries in Metric Spaces.
\newblock  Proceedings of the 16th {ACM SIGACT-SIGMOD-SIGART} Symposium on
  Principles of Database Systems ({PODS'97}); ACM Press: Seattle, WA,  1998;
  pp. 59--68.

\bibitem[Motwani \em{et~al.}(2007)Motwani, Naor, and Panigrahy]{MotwaniNP07}
Motwani, R.; Naor, A.; Panigrahy, R.
\newblock Lower Bounds on Locality Sensitive Hashing.
\newblock {\em SIAM J. Discrete Math.} {\bf 2007}, {\em 21},~930--935.

\bibitem[Paulev{\'e} \em{et~al.}(2010)Paulev{\'e}, J{\'e}gou, and
  Amsaleg]{Pauleve:2010}
Paulev{\'e}, L.; J{\'e}gou, H.; Amsaleg, L.
\newblock Locality Sensitive Hashing: A Comparison of Hash Function Types and
  Querying Mechanisms.
\newblock {\em Pattern Recogn. Lett.} {\bf 2010}, {\em 31},~1348--1358.

\bibitem[Indyk and Motwani(1998)]{Indyk98}
Indyk, P.; Motwani, R.
\newblock Approximate Nearest Neighbors: Towards Removing the Curse of
  Dimensionality.
\newblock  Proceedings of the Thirtieth Annual ACM Symposium on Theory of
  Computing; ACM: New York, NY, USA,  1998; STOC '98, pp. 604--613.

\bibitem[Gionis \em{et~al.}(1999)Gionis, Indyk, and Motwani]{Gionis}
Gionis, A.; Indyk, P.; Motwani, R.
\newblock Similarity Search in High Dimensions via Hashing.
\newblock  Proceedings of the 25th International Conference on Very Large Data
  Bases; Morgan Kaufmann Publishers Inc.: San Francisco, CA, USA,  1999; VLDB
  '99, pp. 518--529.

\bibitem[Joly and Buisson(2008)]{JolyB08}
Joly, A.; Buisson, O.
\newblock A posteriori multi-probe locality sensitive hashing.
\newblock  ACM Multimedia; El-Saddik, A.; Vuong, S.; Griwodz, C.; Bimbo, A.D.;
  Candan, K.S.; Jaimes, A., Eds. ACM,  2008, pp. 209--218.

\bibitem[Datar \em{et~al.}(2004)Datar, Immorlica, Indyk, and
  Mirrokni]{Datar:2004}
Datar, M.; Immorlica, N.; Indyk, P.; Mirrokni, V.S.
\newblock Locality-sensitive Hashing Scheme Based on P-stable Distributions.
\newblock  Proceedings of the Twentieth Annual Symposium on Computational
  Geometry; ACM: New York, NY, USA,  2004; SCG '04, pp. 253--262.

\bibitem[Andoni and Indyk(2008)]{Andoni:2008}
Andoni, A.; Indyk, P.
\newblock Near-optimal Hashing Algorithms for Approximate Nearest Neighbor in
  High Dimensions.
\newblock {\em Commun. ACM} {\bf 2008}, {\em 51},~117--122.

\bibitem[Charikar(2002)]{Charikar:2002}
Charikar, M.S.
\newblock Similarity Estimation Techniques from Rounding Algorithms.
\newblock  Proceedings of the Thiry-fourth Annual ACM Symposium on Theory of
  Computing; ACM: New York, NY, USA,  2002; STOC '02, pp. 380--388.

\bibitem[Indyk(2000)]{Indyk_asublinear}
Indyk, P.
\newblock A Sublinear Time Approximation Scheme for Clustering in Metric
  Spaces.
\newblock  Proc. 40th IEEE FOCS,  2000, pp. 154--159.

\bibitem[Leskovec \em{et~al.}(2010)Leskovec, Lang, and Mahoney]{LeskovecLM10}
Leskovec, J.; Lang, K.J.; Mahoney, M.W.
\newblock Empirical comparison of algorithms for network community detection.
\newblock  WWW; Rappa, M.; Jones, P.; Freire, J.; Chakrabarti, S., Eds. ACM,
  2010, pp. 631--640.

\bibitem[Yang and Leskovec(2012)]{YangL12}
Yang, J.; Leskovec, J.
\newblock Defining and Evaluating Network Communities Based on Ground-Truth.
\newblock  ICDM; Zaki, M.J.; Siebes, A.; Yu, J.X.; Goethals, B.; Webb, G.I.;
  Wu, X., Eds. IEEE Computer Society,  2012, pp. 745--754.

\bibitem[van Dongen(2000)]{dongen}
van Dongen, S.
\newblock {A Cluster Algorithm For Graphs}.
\newblock Technical Report INS-R 0010, CWI, Amsterdam, the Netherlands,  2000.

\bibitem[Eckmann and Moses(2002)]{Eckmann2002}
Eckmann, J.P.; Moses, E.
\newblock {Curvature of co-links uncovers hidden thematic layers in the World
  Wide Web}.
\newblock {\em PNAS} {\bf 2002}, {\em 99},~5825--5829.

\bibitem[Girvan and Newman(2002)]{girvan02}
Girvan, M.; Newman, M.E.J.
\newblock Community structure in social and biological networks.
\newblock {\em Proceedings of the National Academy of Sciences} {\bf 2002},
  {\em 99},~7821--7826.

\bibitem[Zhou and Lipowsky(2004)]{ZhouL04}
Zhou, H.; Lipowsky, R.
\newblock Network Brownian Motion: A New Method to Measure Vertex-Vertex
  Proximity and to Identify Communities and Subcommunities.
\newblock  International Conference on Computational Science; Bubak, M.; van
  Albada, G.D.; Sloot, P.M.A.; Dongarra, J., Eds. Springer,  2004, Vol. 3038,
  {\em Lecture Notes in Computer Science}, pp. 1062--1069.

\bibitem[Reichardt and Bornholdt(2004)]{spinglass}
Reichardt, J.; Bornholdt, S.
\newblock Detecting fuzzy community structures in complex networks with a Potts
  model.
\newblock {\em Phys Rev Lett} {\bf 2004}, {\em 93},~218701.

\bibitem[Clauset \em{et~al.}(2004)Clauset, Newman, , and Moore]{Clauset2004}
Clauset, A.; Newman, M.E.J.; .; Moore, C.
\newblock Finding community structure in very large networks.
\newblock {\em Physical Review E} {\bf 2004}, pp. 1-- 6.

\bibitem[Wu and Huberman(2004)]{linear}
Wu, F.; Huberman, B.
\newblock Finding communities in linear time: a physics approach.
\newblock {\em The European Physical Journal B - Condensed Matter and Complex
  Systems} {\bf 2004}, {\em 38},~331--338.

\bibitem[Fortunato \em{et~al.}(2004)Fortunato, Latora, and
  Marchiori]{infocentrality}
Fortunato, S.; Latora, V.; Marchiori, M.
\newblock Method to find community structures based on information centrality.
\newblock {\em Physical Review E (Statistical, Nonlinear, and Soft Matter
  Physics)} {\bf 2004}, {\em 70},~056104.

\bibitem[Donetti and Muñoz(2004)]{Donetti}
Donetti, L.; Muñoz, M.A.
\newblock Detecting network communities: a new systematic and efficient
  algorithm.
\newblock {\em Journal of Statistical Mechanics: Theory and Experiment} {\bf
  2004}, {\em 2004},~P10012.

\bibitem[Guimera and Amaral(2005)]{Guimera04}
Guimera, R.; Amaral, L.A.N.
\newblock Functional cartography of complex metabolic networks.
\newblock {\em Nature} {\bf 2005}, {\em 433},~895--900.

\bibitem[Capocci \em{et~al.}(2004)Capocci, Servedio, Caldarelli, and
  Colaiori]{Capocci04}
Capocci, A.; Servedio, V.D.P.; Caldarelli, G.; Colaiori, F.
\newblock {Detecting communities in large networks}.
\newblock {\em Physica A: Statistical Mechanics and its Applications} {\bf
  2004}, {\em 352},~669--676.

\bibitem[Duch and Arenas(2005)]{duch05}
Duch, J.; Arenas, A.
\newblock Community detection in complex networks using Extremal Optimization.
\newblock {\em Physical Review E} {\bf 2005}, {\em 72},~027104.

\bibitem[Bagrow and Bollt(2005)]{Bagrow05}
Bagrow, J.P.; Bollt, E.M.
\newblock Local method for detecting communities.
\newblock {\em Phys. Rev. E} {\bf 2005}, {\em 72},~046108.

\bibitem[Palla \em{et~al.}(2005)Palla, Derenyi, Farkas, and Vicsek]{palla05}
Palla, G.; Derenyi, I.; Farkas, I.; Vicsek, T.
\newblock {Uncovering the overlapping community structure of complex networks
  in nature and society}.
\newblock {\em Nature} {\bf 2005}, {\em 435},~814--818.

\bibitem[Raghavan \em{et~al.}(2007)Raghavan, Albert, and Kumara]{raghavan}
Raghavan, U.N.; Albert, R.; Kumara, S.
\newblock Near linear time algorithm to detect community structures in
  large-scale networks.
\newblock {\em Phys. Rev. E} {\bf 2007}, {\em 76},~036106.

\bibitem[Rosvall and Bergstrom(2008)]{rosvall2008random}
Rosvall, M.; Bergstrom, C.T.
\newblock Maps of random walks on complex networks reveal community structure.
\newblock {\em Proceedings of the National Academy of Sciences} {\bf 2008},
  {\em 105},~1118--1123.

\bibitem[Ronhovde and Nussinov(2009)]{Ronhovde09}
Ronhovde, P.; Nussinov, Z.
\newblock Multiresolution community detection for megascale networks by
  information-based replica correlations.
\newblock {\em Phys. Rev. E} {\bf 2009}, {\em 80},~016109.

\bibitem[Leskovec \em{et~al.}(2008)Leskovec, Lang, Dasgupta, and
  Mahoney]{Leskovec2008}
Leskovec, J.; Lang, K.J.; Dasgupta, A.; Mahoney, M.W.
\newblock Community Structure in Large Networks: Natural Cluster Sizes and the
  Absence of Large Well-Defined Clusters,  2008.

\bibitem[Gonzalez(1985)]{Gonzalez85}
Gonzalez, T.F.
\newblock Clustering to Minimize the Maximum Intercluster Distance.
\newblock {\em Theor. Comput. Sci.} {\bf 1985}, {\em 38},~293--306.

\end{thebibliography}


%


%

\end{document}